\documentclass{article}

\usepackage{arxiv}

\usepackage[utf8]{inputenc} 
\usepackage[T1]{fontenc}    
\usepackage{hyperref}       
\usepackage{url}            
\usepackage{booktabs}       
\usepackage{amsfonts}       
\usepackage{nicefrac}       
\usepackage{microtype}      
\usepackage{lipsum}

\usepackage{cite}
\usepackage{amsmath,amssymb,amsfonts}
\usepackage{graphicx}
\usepackage{algorithm}
\usepackage[noend]{algpseudocode}
\usepackage{hyperref}
\usepackage{textcomp}
\usepackage{threeparttablex} 

\usepackage{mathtools}
\usepackage{caption}
\usepackage{float}
\usepackage{stmaryrd}
\usepackage{epstopdf}
\usepackage{multirow}
\usepackage{pifont}
\usepackage{caption}
\usepackage{subcaption}

\newcommand{\R}{{\mathbb{R}}}

\newcommand{\Cs}{{\mathcal{C}}}
\newcommand{\N}{{\mathbb{N}}}

\newcommand{\X}{{\mathbf{X}}}

\newcommand{\T}{{\mathbf{T}}}
\newcommand{\So}{{\mathbf{S}}}
\newcommand{\Obs}{{\mathcal{U}}}

\newcommand{\U}{{\mathbf{U}}}

\newcommand{\cmark}{\ding{51}}%
\newcommand{\xmark}{\ding{55}}%

\newtheorem{theorem}{Theorem}[section]

\newtheorem{definition}[theorem]{Definition}

\newtheorem{remark}[theorem]{Remark}
\newtheorem{problem}[theorem]{Problem}
\newtheorem{proof}[theorem]{Proof}

\title{Control Barrier Functions for Prescribed-time Reach-Avoid-Stay Tasks using Spatiotemporal Tubes
\thanks{ This work was supported in part by the SERB Start-Up Research Grant; in part by the ARTPARK. The work of Ratnangshu Das was supported by the Prime Minister’s Research Fellowship from the Ministry of Education, Government of India.}
}

\author{
 Ratnangshu Das \\
  Robert Bosch Centre for Cyber-Physical Systems\\
  IISc, Bengaluru, India\\
  \texttt{ratnangshud@iisc.ac.in} \\
   \And
 Pranav Bakshi \\
  Mechanical Engineering Department\\
  IIT, Kharagpur, India\\
  \texttt{pranavbakshirg4@kgpian.iitkgp.ac.in} \\
  \And
 Pushpak Jagtap \\
  Robert Bosch Centre for Cyber-Physical Systems\\
  IISc, Bengaluru, India\\
  \texttt{pushpak@iisc.ac.in} \\
}

\begin{document}
\maketitle

\begin{abstract}
Prescribed-time reach-avoid-stay (PT-RAS) specifications are crucial in applications requiring precise timing, state constraints, and safety guarantees. While control carrier functions (CBFs) have emerged as a promising approach, providing formal guarantees of safety, constructing CBFs that satisfy PT-RAS specifications remains challenging. In this paper, we present a novel approach using a spatiotemporal tubes (STTs) framework to construct CBFs for PT-RAS tasks. The STT framework allows for the systematic design of CBFs that dynamically manage both spatial and temporal constraints, ensuring the system remains within a safe operational envelope while achieving the desired temporal objectives. 
The proposed method is validated with two case studies: temporal motion planning of an omnidirectional robot and temporal waypoint navigation of a drone with obstacles, using higher-order CBFs.
\end{abstract}

\section{Introduction}

Prescribed-time reach-avoid-stay (PT-RAS) specifications are critical in applications that involve guiding a system to reach a desired state within a specified time, avoiding unsafe regions, and respecting state constraints \cite{Meng1}. PT-RAS tasks also serve as fundamental building blocks in the design of complex specifications \cite{Kloetzer, NAHS} for autonomous systems involving temporal and spatial constraints. 
Effective design control strategies ensuring PT-RAS task is crucial in applications like robotics, autonomous vehicles, and aerospace to ensure reliability and safety with precise timing.

Several control techniques have been proposed in the literature to address these specifications, including model predictive control (MPC) \cite{MPC} and potential field methods \cite{APF,APF_book}. While these approaches can handle time-bound tasks and obstacle avoidance, they often suffer from difficulty in ensuring safety guarantees over the entire mission duration. These limitations highlight the need for more efficient and reliable methods that can provide formal safety guarantees. Symbolic control techniques \cite{tabuada2009verification} have emerged as powerful tools for specifying and solving complex tasks. However, these techniques typically rely on state space abstraction, which can lead to increased computational complexity. On the other hand, nonlinear control methods, such as the funnel-based control approach \cite{PPC1}, offer a computationally efficient means of achieving specific tracking performance. Despite this, they struggle with addressing non-convex tasks like obstacle avoidance \cite{Funnel_STL} and handling input constraints.

Control barrier functions (CBFs) \cite{CBF,jagtap2020formal} have gained significant attention as a promising solution to some of these challenges. CBF-based techniques provide formal guarantees that a system will avoid unsafe states \cite{CBF_TA} while adhering to input constraints \cite{jagtap2020formal}. In the context of RAS specifications, several studies have integrated CBFs with well-known path-planning algorithms like Rapidly Exploring Random Trees (RRT) \cite{CBF-RRT1, CBF-RRT2, CBF-RRT3} and Artificial potential fields (APFs) \cite{CBF-APF} to generate safe motion plans in obstacle-laden environments. Additionally, the application of control (Lyapunov-) barrier functions \cite{Meng3, C3BF} for controller synthesis has been demonstrated to meet more general RAS specifications. 

However, the underlying bottleneck in ensuring safety using CBFs is the construction of these functions, particularly under nonlinear or non-convex constraints \cite{learn-better-CBF}, which are often inherent to PT-RAS specifications. Recent approaches have attempted to address these challenges by framing the CBF synthesis task as a machine-learning problem. From using supervised safe and unsafe data samples \cite{learn-CBF-1} to data from the expert demonstrations \cite{learn-CBF-ED}, data-driven learning techniques have been readily used to synthesize CBF. However, some methods lack formal correctness guarantees, while others suffer from increased computation complexity.

Further, designing CBFs that adhere to PT-RAS specifications presents a non-trivial challenge. The complexity of simultaneously managing the spatial and temporal aspects of PT-RAS tasks often leads to difficulties in formulating appropriate barrier functions that can dynamically adapt to the evolving state of the system. Time-varying CBFs \cite{TVCBF} have been identified as particularly useful for addressing PT-RAS specifications, as they adjust their constraints over time, allowing for the forward invariance of a time-varying set.

Moving away from learning-based techniques, in this paper, we propose a novel approach to address these challenges by leveraging the spatiotemporal tube (STT) framework \cite{STT} for the design of CBFs. The STT framework enables the systematic design of CBFs that ensure the system remains within a safe operational envelope while meeting the required temporal objectives. To illustrate the effectiveness of this approach, we present two case studies: an omnidirectional robot path planning task in 2D space using CBFs and temporal waypoint navigation of a drone in a 3D space with obstacles, employing higher-order CBFs (HO-CBFs) \cite{HOCBF}. 

\section{Preliminaries and Problem Formulation}
\label{sec:prelim}
\subsection{Notation}
The symbols $\N$, $\N_0$, $ \R$, $\R^+$, and $\R_0^+ $ denote the set of natural, whole, real, positive real, and nonnegative real numbers, respectively. 
A vector space of real matrices with $ n $ rows and $ m $ columns is denoted by  $ \R^{n\times m} $. A column vector with $n$ rows is represented by $ \R^{n}$.
A vector $x \in \mathbb{R}^{n}$ with entries $x_1, \ldots, x_n$ is represented as $[x_1, \ldots, x_n]^\top$, where $x_i \in \mathbb{R}$ denotes the $i$-th element of vector $x\in\mathbb{R}^n$ and $i \in [1;n]$.
Given a matrix $M\in\R^{n\times m}$, $M^\top$ represents the transpose of matrix $M$. 
Given $N \in \N$ sets $\X_i$, $i\in\left[1;N\right]$, we denote the Cartesian product of the sets by $\X=\prod_{i\in\left[1;N\right]}\X_i:=\{(x_1,\ldots,x_N)|x_i\in \X_i,i\in\left[1;N\right]\}$.

\subsection{System Definition}
We consider the following continuous-time nonlinear time-varying affine control system $\mathcal{S}$:
\begin{align}
    \mathcal{S}: \dot{x} = f(t,x) + g(t,x)u,
    \label{eq:sysdyn}
\end{align}
where $ x(t) \in \X \subset \mathbb{R}^n $ is the state of the system and $ u(t) \in \mathbb{R}^m $ is the control input to the system at time $t$. The compact set $\X$ defines the system's state space. The functions $ f: \R_0^+ \times \X \rightarrow \mathbb{R}^n $ and $ g:\R_0^+ \times \X \rightarrow \mathbb{R}^{n \times m} $ are locally Lipschitz continuous.


\subsection{System Specification}
Let $\U = \bigcup_{j \in [1;n_u]} \Obs^j \subset \X$ be an unsafe set, where $n_u \in \N_0$ and $\Obs^j \subset \X$ is assumed to be compact and convex. Note that, in general, $\U$ can be nonconvex and disconnected. The compact connected sets $\So \subset \X \setminus \U$ and $\T \subset \X  \setminus \U$ represent the initial and target sets, respectively. If $\X$ is of any arbitrary shape, redefine the state space as $\hat{\X}:= \prod_{i \in [1;n]} [\underline{\X}_i, \overline{\X}_i]$ and expand the unsafe set as $\hat{\U} = \U \cup (\hat{\X} \setminus \X)$. Consider a prescribed-time reach-avoid-stay (PT-RAS) task defined in Definition \ref{def:ptras}.

\begin{definition}[Prescribed-Time Reach-Avoid-Stay Task]\label{def:ptras}
    Given a system $\mathcal{S}$, a state space $\X$, an unsafe set $\U$, an initial set $\So$, a target set $\T$, and a given initial position $x(0) \in \So$, there exists $t \in [0,t_c]$, s.t., $x(t) \in \T$ and for all $s \in [0,t_c], x(s) \in \X \setminus \U$, where $t_c \in \R^+$ is the prescribed-time.
\end{definition}

\subsection{Time-varying Control Barrier Functions}
\label{subsec:desc}
Time-varying control barrier functions (TV-CBFs) can address these PT-RAS specifications by introducing a flexible yet robust framework that dynamically adapts to changes over time.
Following the approach in \cite{TVCBF}, we introduce a TV-CBF, defined by a differentiable function $b: {\left[0, t_c\right]} \times \X \rightarrow \mathbb{R}$ with $\X \in \mathbb{R}^n$. The time-dependent safe set $\Cs(t) \subseteq \X \subset \R^n$ is defined as the 0-superlevel set of $b(t,x)$, yielding
\begin{subequations}\label{eq:safeset}
\begin{align}
    \Cs(t) &= \{x \in \X | b(t,x) \geq 0\} \\
    \partial\Cs(t) &= \{x \in \X | b(t,x) = 0\} \\
    \text{Int}(\Cs(t)) &= \{x \in \X | b(t,x) > 0\},
\end{align}
\end{subequations}
where $\partial \Cs$ and Int$(\Cs)$ are the boundary and interior of the set  $\Cs$, respectively.
Next,  we present the conditions under which $b(t,x)$ fits as a candidate control barrier function.

\begin{definition}[Candidate Control Barrier Function \cite{CBF_STL}]
A differentiable function $ b : \left[0, t_c\right] \times \X \rightarrow \R $ is considered a candidate control barrier function if, for all $x_0 \in \Cs(0)$, there exists an absolutely continuous function $ x : [0, t_c] \rightarrow \R^n $ with $ x(0) := x_0 $ such that $ x(t) \in \Cs(t) $ for all $ t \in \left[0, t_c\right] $.
\end{definition}

\begin{definition}[Forward Invariant Set \cite{HOCBF}]
A set $\Cs(t)$ $\subseteq \X$ is forward invariant for system \eqref{eq:sysdyn} if starting from any $x(0) \in \Cs(0)$ the system's solution $ x(t) \in \Cs(t) $ for all $ t \in \left[0, t_c\right] $.
\end{definition}

For nonlinear affine systems (\ref{eq:sysdyn}), a candidate CBF is considered a valid CBF if it can be used to guarantee the forward invariance of the set $\Cs(t)$. 

\begin{definition}[Valid Control Barrier Function]
    Consider a continuous-time control-affine system $\mathcal{S}$ in \eqref{eq:sysdyn} and a safe set $\Cs(t) \subseteq \X$ in \eqref{eq:safeset} as a superlevel set of a continuous function $b: {\left[0, t_c\right]} \times \X \rightarrow \mathbb{R}$. The set $\Cs(t)$ is a controlled invariant set for the system $\mathcal{S}$ in \eqref{eq:sysdyn}, if there exists a class $\mathcal{K}$ function, $\alpha: \R \rightarrow \R$, such that, for all $ (t,x) \in [0, t_c] \times \X $, 
    \begin{equation}
    \label{eq:CBF}
    \sup_{u} \left( \frac{\partial b(t,x)}{\partial x}^\top (f(t,x) + g(t,x)u) + \frac{\partial b(t,x)}{\partial t} \right) \geq -\alpha(b(t,x)).
    \end{equation}    
\end{definition}

However, given a PT-RAS task, coming up with the time-varying CBF $b(t,x)$ is particularly challenging. Next, we formally define the problem considered in this work.

\begin{problem}\label{prob1}
   Given the system $\mathcal{S}$ in \eqref{eq:sysdyn} and a PT-RAS task as defined in Definition \ref{def:ptras}, the objective is to design a TV-CBF, $b: {\left[0, t_c\right]} \times \X \rightarrow \mathbb{R}$, that guarantees the system's safe operation within the prescribed time frame $[0,t_c]$. Consequently, derive a control law $u(t,x)$ such that the solution $x: \R_0^+ \rightarrow \R^n$ to \eqref{eq:sysdyn}, starting from any initial condition $x(0) \in \So$, satisfies the PT-RAS specifications.
\end{problem}

\section{STT-based Construction of TV-CBF}
This section introduces the mathematical framework STT-CBF, to generate TV-CBFs using spatiotemporal tubes (STTs).

\subsection{Spatiotemporal Tubes (STT)} 
To design the time-varying CBF, this paper employs the concept of STTs \cite{STT}. The STTs serve as dynamic structures, effectively capturing the system's evolving spatial constraints and temporal requirements. They offer a robust solution to designing TV-CBFs that direct the system toward its target within the prescribed time while maintaining safety by avoiding unsafe sets.

\begin{definition}[Spatiotemporal Tubes for PT-RAS]\label{def:stt}
Given a PT-RAS task in Definition \ref{def:ptras}, the STT is defined as a time-varying interval $[\gamma_{i,L}(t), \gamma_{i,U}(t)]$, where $\gamma_{i,L}:\R_0^+ \rightarrow \R$ and $\gamma_{i,U}:\R_0^+ \rightarrow \R$ are continuously differentiable functions with $\gamma_{i,L}(t) < \gamma_{i,U}(t)$ for all time $t \in [0,t_c]$ and for each dimension $i \in [1;n]$, and the following holds:
\begin{align}\label{eqn:sttd}
    &\prod_{i=1}^n [\gamma_{i,L}(t), \gamma_{i,U}(t)] \subseteq \X, \forall t \in [0,t_c], \nonumber \\
    &\prod_{i=1}^n [\gamma_{i,L}(0), \gamma_{i,U}(0)] \subseteq \So, \
    \prod_{i=1}^n [\gamma_{i,L}(t_c), \gamma_{i,U}(t_c)] \subseteq \T, \nonumber \\
    &\prod_{i=1}^n 
    [\gamma_{i,L}(t), \gamma_{i,U}(t)] \cap \U = \emptyset, \forall t \in [0,t_c].
\end{align}
\end{definition}
As discussed in \cite{STT}, three steps are primarily involved in the generation of spatiotemporal tubes, (i) designing reachability tubes to guide the system trajectory from the initial set $\So$ to the target set $\T$, (ii) introducing the circumvent function to avoid the unsafe set $\U$, and (iii) dynamically adjusting the tubes around the unsafe region through the adaptive framework.

Thus, we can enforce the PT-RAS specification by constraining the state trajectory within the spatiotemporal tubes as:
\begin{gather}
    \gamma_{i,L}(t) < x_i(t) < \gamma_{i,U}(t), \forall (t,i) \in [0,t_c] \times [1;n]. \label{eqn:ppc}
\end{gather}

\subsection{STT-CBF framework}
To construct a TV-CBF utilizing the STT framework, we first define the normalized time-varying error within the tubes as $e(t,x) = [e_1(t, x_1), \ldots, e_n(t, x_n)]$, where
\begin{equation} \label{eq:normerr}
e_i(t, x_i) = \frac{x_i - \frac{1}{2}(\gamma_{i,U}(t) + \gamma_{i,L}(t))}{\frac{1}{2}(\gamma_{i,U}(t) - \gamma_{i,L}(t))}, \forall i \in [1;n]   
\end{equation}
is the error in $i$th dimension. Note that the error $e_i(t, x_i)$ satisfies the following statements:
\begin{itemize}
    \item $e_i(t, x_i) \in (-1, 1)$ if the $i$th position is between the lower and upper boundaries of the tube in $i$th dimension, i.e.,
    $$e_i(t, x_i) \in (-1, 1), \text{ iff } x_i(t) \in (\gamma_{i,L}(t), \gamma_{i,U}(t))$$
    \item $e_i(t, x_i)$ approaches $-1$ and $1$ as the $x_i$ approaches the lower and upper boundary of the tube in $i$th dimension, respectively, i.e., 
    $$\lim_{x_i(t) \rightarrow \gamma_{i,L}(t)}e_i(t, x_i) = -1, \
    \lim_{x_i(t) \rightarrow \gamma_{i,U}(t)}e_i(t, x_i) = 1.$$
\end{itemize}

Now, we define our candidate TV-CBF $b(t,x)$ as
\begin{equation}
b(t,x) = 1 - \frac{1}{a_1} \ln \left(\sum_{i=1}^{n} e^{a_1 e_i(t, x_i) \tanh(a_2 e_i(t, x_i))}\right), a_1, a_2 \in \mathbb{R}^+,
\label{eq:barrier_function}
\end{equation}
where $e_i(t, x_i)$ is the normalized error in \eqref{eq:normerr}, and satisfies the following property,
\begin{equation*}
    b(t,x) = 
    \begin{cases} 
    \geq 0 & \text{iff } e_i(t, x_i) \in (-1, 1), \forall i \in [1;n] \\
    < 0 & \text{otherwise.}
    \end{cases}    
\end{equation*}


\begin{theorem}
Consider the TV-CBF $b(t,x)$ in \eqref{eq:barrier_function} and its 0-superlevel set $\Cs(t)$ as defined in \eqref{eq:safeset}. If the system's initial state satisfies $x(0) \in \Cs(0)$ and $\Cs(t)$ is forward invariant under the dynamics in \eqref{eq:sysdyn}, then the system satisfies the PT-RAS task.
\end{theorem}
\begin{proof}
The 0-superlevel set $\Cs(t)$ represents the set of safe states at time $t$, defined as the points in the state space where the barrier function $b(t,x)$ is non-negative: 
$$x(t) \in \Cs(t) \implies b(t,x) \geq 0.$$
From \eqref{eq:normerr}, it follows that for each state variable $x_i(t)$, if $b(t,x) \geq 0$ then the error $e_i(t,x_i)$ lies within the interval $(-1,1)$. Consequently, the state variable $x_i(t)$ is confined to the interval $(\gamma_{i,L}(t), \gamma_{i,U}(t))$ for each $i \in [1;n]$:
\begin{align*}
    b(t,x) \geq 0, &\implies e_i(t, x_i) \in (-1, 1) \implies x_i(t) \in (\gamma_{i,L}(t), \gamma_{i,U}(t)), \forall i \in [1;n].    
\end{align*}
Therefore, if the system's trajectory remains within the superlevel set $\Cs(t)$, the state $x(t)$ is guaranteed to lie within the hyperrectangle bounded by the spatiotemporal tubes (STTs) at time $t$:
$$x(t) \in \Cs(t) \implies x(t) \in \prod_{i=1}^{n}(\gamma_{i,L}(t), \gamma_{i,U}(t)).$$ 
Given that $\Cs(t)$ is forward invariant under the system dynamics \eqref{eq:sysdyn}, we have:
$$\forall x(0) \in \Cs(0), x(t) \in \Cs(t) \subseteq \prod_{i=1}^n (\gamma_{i,L}(t), \gamma_{i,U}(t)).$$ 
Thus, the trajectory $x(t)$ is enclosed within the STT, ensuring the PT-RAS task is satisfied.
\end{proof}

\begin{theorem}
The time-varying control barrier function $b(t,x)$, defined in \eqref{eq:barrier_function}, is differentiable with respect to $t$ and $x$.
\end{theorem}
\begin{proof}
By Definition \ref{def:stt}, the bounds of the spatiotemporal tubes, $\gamma_L(t)$ and $\gamma_U(t)$, are continuously differentiable functions of time $t$. Consequently, from \eqref{eq:normerr}, as $\gamma_{i,U}(t) > \gamma_{i,L}(t), \forall (t,i) \in \R_0^+ \times [1;n]$, the normalized error $e(t, x)$ is differentiable with respect to both $t$ and $x$.
Since $b(t,x)$ is a composition of differentiable functions \eqref{eq:barrier_function}, it is also differentiable with respect to both $t$ and $x$.
\end{proof}

To illustrate the construction of TV-CBF from STTs, consider a two-dimensional example in Figure \ref{fig:bf}. The figure shows the 0-superlevel set $\Cs(t)$ of $b(t,x)$ at a certain time instance $t$, representing the region in the state space where the barrier function $b(t,x)$ is non-negative, indicating the set of safe states at that moment.

We can further observe that $\Cs(t)$ is a subset of the hyperrectangle defined by the STTs at time $t$. This containment guarantees that if the system's state remains within the safe set $\Cs(t)$, it simultaneously adheres to the tube constraints, thereby ensuring the satisfaction of the PT-RAS specifications. 
\begin{figure}[h!]
    \centering
    \includegraphics[width=0.8\textwidth]{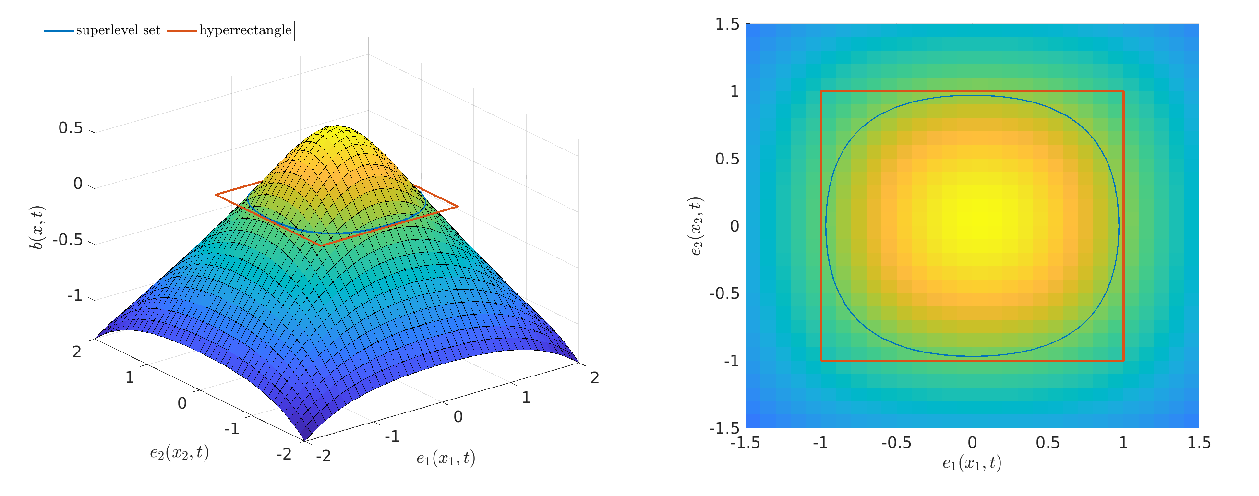}
    \caption{Barrier function for a two dimensional case with $a_1 = a_2 = 2$. At a given time instance $t$, the superlevel set of $b(t,x)$, $\Cs(t)$ is a subset of the hyperrectangle enclosed by the tubes.}
    \label{fig:bf}
\end{figure}

\begin{figure*}[t]
     \centering
     \begin{subfigure}[b]{\textwidth}
         \centering
         \includegraphics[width=\textwidth]{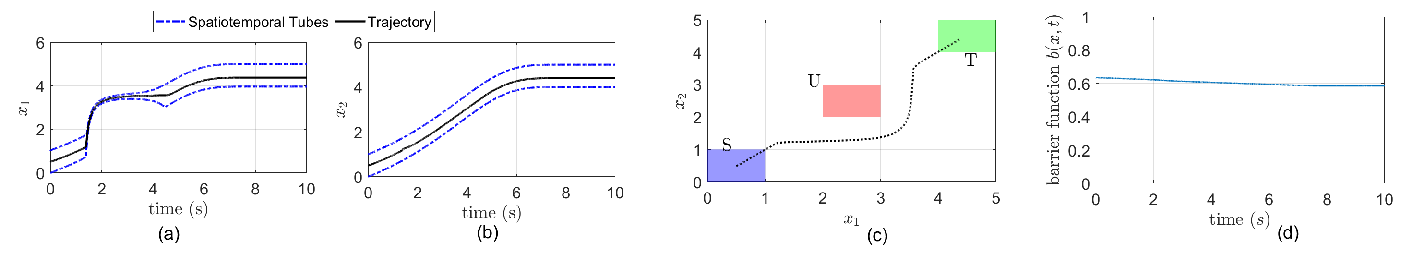}
     \end{subfigure}
     \hfill
     \begin{subfigure}[b]{\textwidth}
         \centering
         \includegraphics[width=\textwidth]{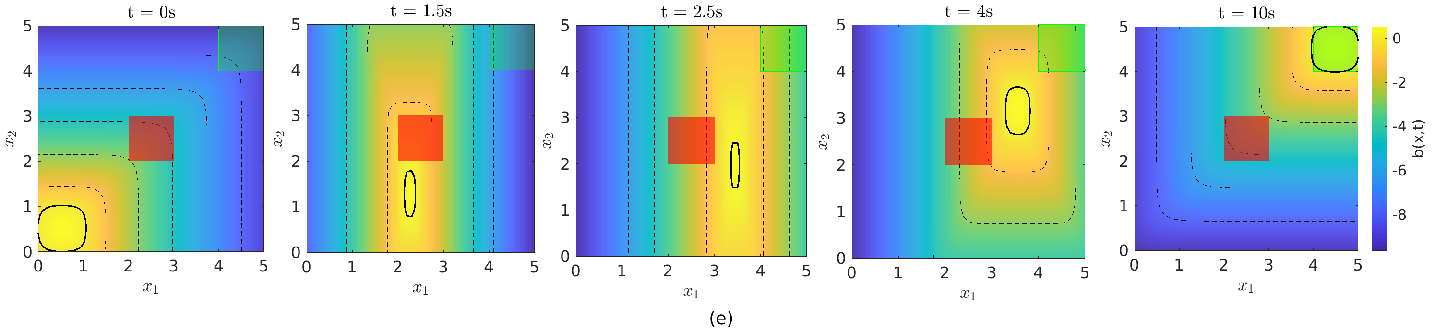}
     \end{subfigure}
        \caption{(a),(b) Presents the spatiotemporal tubes in $x_1$ and $x_2$ dimensions respectively, (c) Controlled trajectory, (d) $b(t,x) > 0$ for all time,(e) Evolution of $b(t,x)$ with time, where the superlevel set $\Cs(t)$ is depicted in black solid line.}
        \label{fig:demoSTTCBF}
\end{figure*}

\subsection{Controller Design}
We now formulate a quadratic programming (QP) problem that ensures $\Cs(t)$ is forward invariant, provided $b(t,x)$ is a valid CBF. 
\begin{align}
&\min_{u}u^T K u \notag\\
\text{s.t. } &-\mathcal{L}_g b(t,x) u \leq \mathcal{L}_fb(t,x) + \frac{\partial b(t,x)}{\partial t} + \alpha(b(t,x)),\label{eq:QP}
\end{align}
where $K \in \R^{n \times n}$ is a positive semi-definite matrix. 

Given a set $\Cs(t)$ associated with the TV-CBF $b(t,x)$ in \eqref{eq:barrier_function}, the control law $u^*(t,x)$ obtained by solving the QP \eqref{eq:QP} renders the set $\Cs(t)$ forward invariant. Consequently, 
$$x(t) \in \Cs(t) \subset \prod_{i=1}^{n} (\gamma_{i,L}(t), \gamma_{i,U}(t)), \forall t \in [0, t_c],$$
which leads to,
\begin{subequations}
\begin{align}
&x(0) \in \Cs(0) \subset \prod_{i=1}^{n} (\gamma_{i,L}(0), \gamma_{i,U}(0)) \subset \So, \\
&x(t_c) \in \Cs(t_c) \subset \prod_{i=1}^{n} (\gamma_{i,L}(t_c), \gamma_{i,U}(t_c)) \subset \T, \\
&\forall t \in [0, t_c], x(t) \subset \X \text{ and } x(t) \cap \U(t) = \emptyset.
\end{align}
\end{subequations}    
Hence, the STT-CBF controller $u^*(t,x)$ guarantees the satisfaction of the PT-RAS specification in Definition \ref{def:ptras}.

Figure \ref{fig:demoSTTCBF} illustrates a two-dimensional case, where the control law $u^*(t,x)$ obtained by solving the QP directs the trajectory towards the target $\T$ while circumventing the unsafe set $\U$. The evolution of the time-varying CBF $b(t,x)$ and its corresponding superlevel set $\Cs(t)$ is also shown in the Figure. 

\begin{remark}
Up to this point, we have focused on the application of the STT-CBF framework for generating a valid CBF in systems where the relative degree is 1. However, the STT-CBF framework also excels in handling complex dynamical systems with relative degrees greater than 1, where control inputs influence higher-order derivatives of the state variables rather than the state itself. This advantage is demonstrated in the second case study, where we apply the STT-CBF controller to a multi-UAV payload system with relative degree $\eta > 1$ — defined as the number of times the barrier function $b(t,x)$ must be differentiated with respect to $x$ along the system dynamics in \eqref{eq:sysdyn} before the control input $u$ explicitly appears.
\end{remark}

\begin{figure*}
    \centering
    \includegraphics[width=\textwidth]{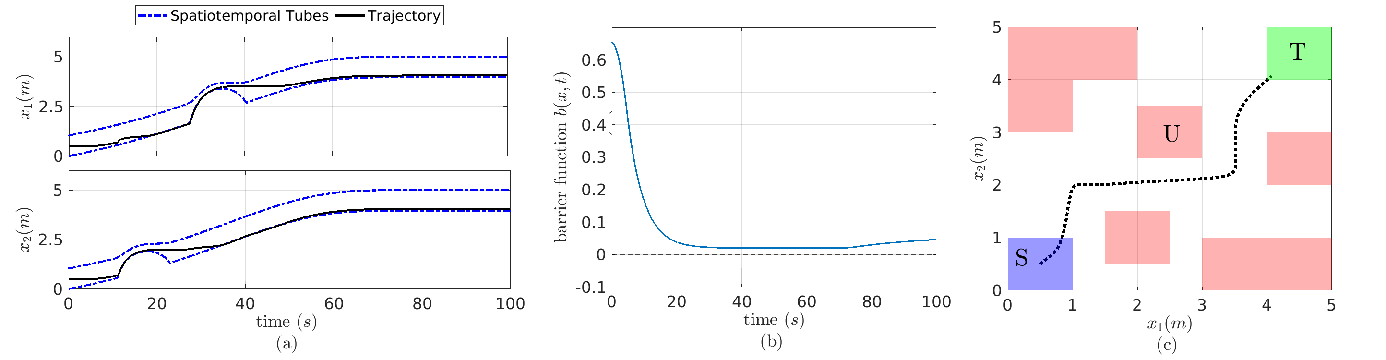}
    \caption{(a) STTs, (b) Barrier function $b(t,x)$, (c) Controlled trajectory, for the mobile robot case}
        \label{fig:cs1}
\end{figure*}

\begin{figure*}
     \centering
     \begin{subfigure}[b]{0.42\textwidth}
         \centering
         \includegraphics[width=\textwidth]{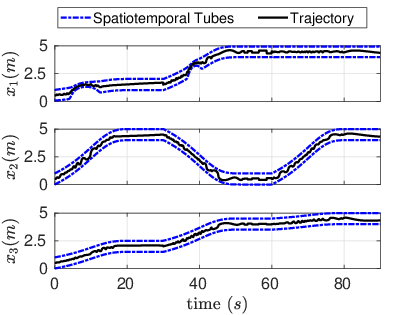}
         \caption{}
     \end{subfigure}
     \hfill
     \begin{subfigure}[b]{0.42\textwidth}
         \centering
         \includegraphics[width=\textwidth]{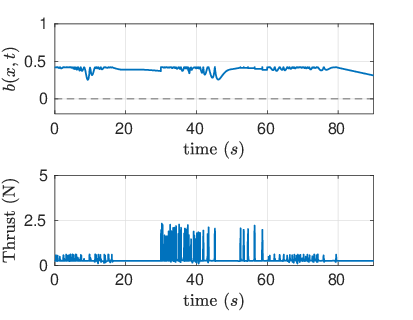}
         \caption{}
     \end{subfigure}
     \hfill
     \begin{subfigure}[b]{0.42\textwidth}
         \centering
         \includegraphics[width=\textwidth]{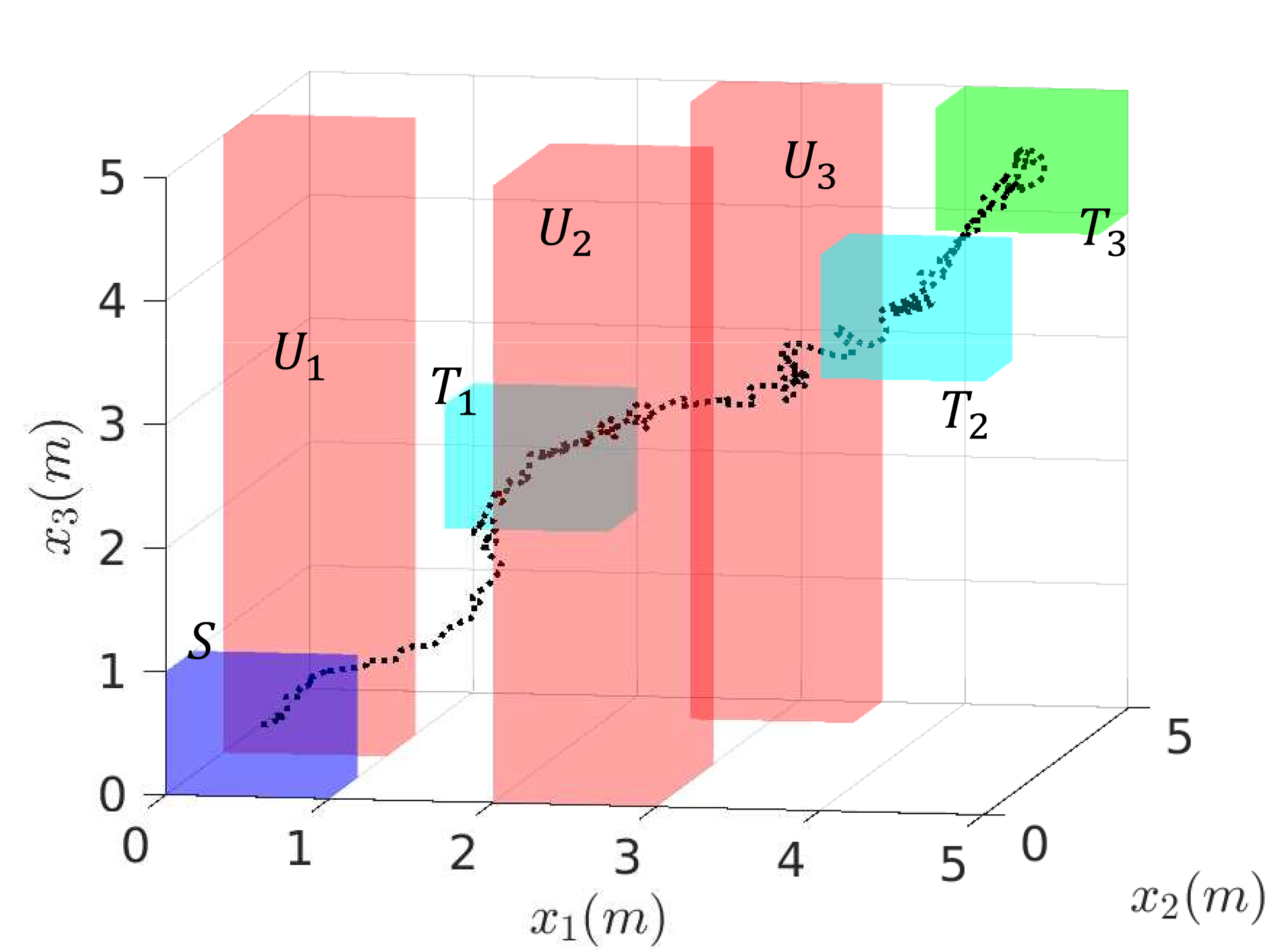}
         \caption{}
     \end{subfigure}
     \hfill
     \begin{subfigure}[b]{0.42\textwidth}
         \centering
         \includegraphics[width=\textwidth]{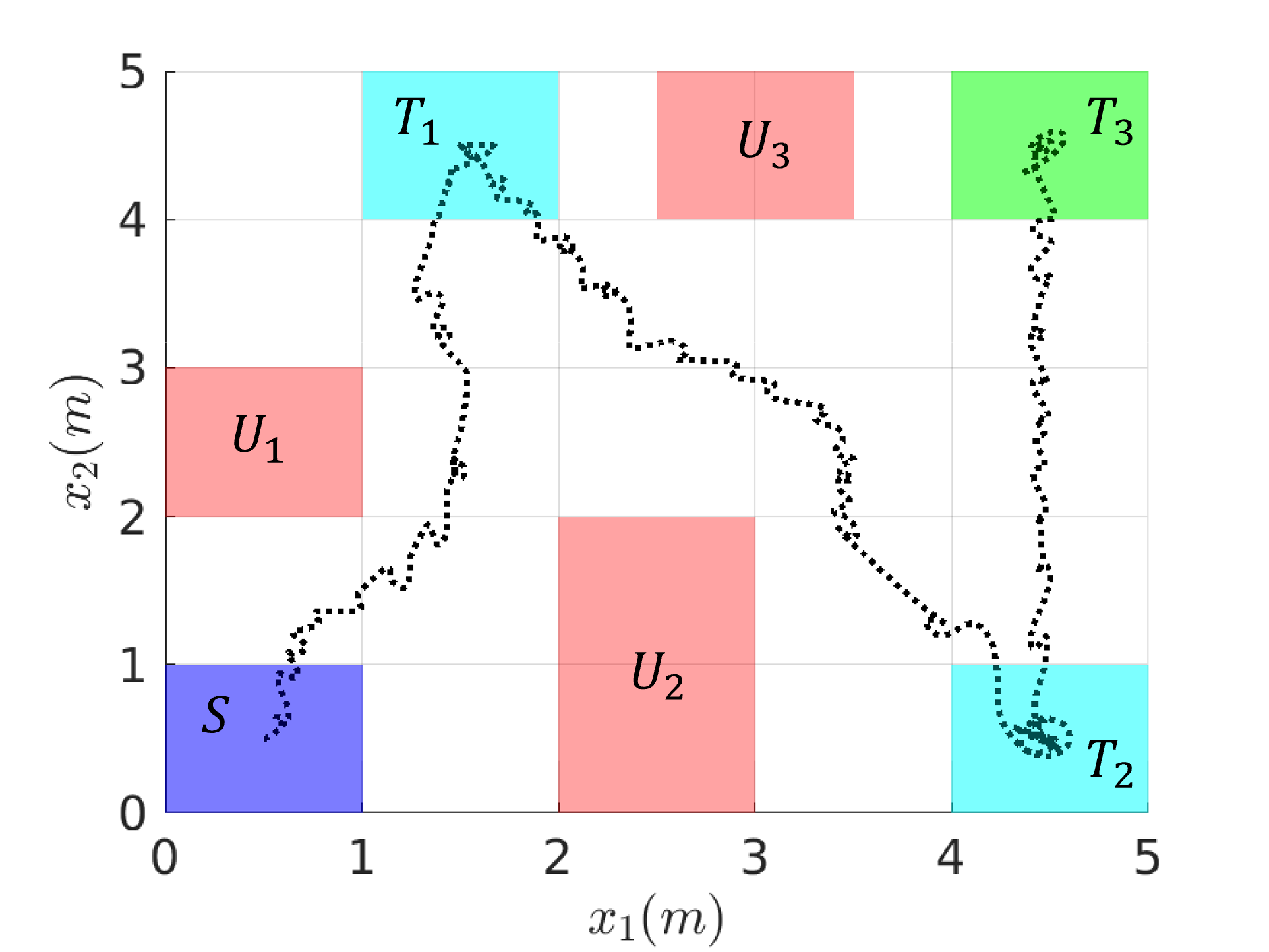}
         \caption{}
     \end{subfigure}
        \caption{(a) STTs, (b) Barrier function $b(t,x)$ and Thrust, (c), (d) Controlled trajectory, for the multi-UAV case.}
        \label{fig:demoSTTCBF_Drone}
\end{figure*}

\section{Case Studies}
We demonstrate the effectiveness of the proposed approach on two different systems: a two-dimensional mobile robot and a multi-UAV payload system operating in a three dimensional environment.
\subsection{Mobile Robot}
In this case study, we consider a mobile robot operating in a 2D environment, as shown in Figure \ref{fig:cs1}(c). The vehicle dynamics have been adopted from \cite{NAHS}:
\begin{align}
    \begin{bmatrix}
        \dot{x}_1 \\ \dot{x}_2 \\ \dot{x}_3
    \end{bmatrix}
    = 
    \begin{bmatrix}
        \cos{x_3} & -\sin{x_3} & 0 \\ \sin{x_3} & \cos{x_3} & 0 \\ 0 & 0 & 1
    \end{bmatrix}
    \begin{bmatrix}
        v_1 \\ v_2 \\ \omega
    \end{bmatrix} + w(t),
\end{align}
where the state vector $x = [x_1, x_2, x_3]^\top$ represents the robot's position and orientation, $[v_1, v_2, \omega]^\top$ is the velocity input in the robot's frame, and $w$ is an unknown disturbance.

The STT-CBF controller ensures that the system trajectory initiating at $x(0) \in \So := [0,1] \times [0,1]$, reaches the target set $\T = [4,5] \times [4,5]$, within prescribed time $t_c = 8s$, while avoiding unsafe set $\U$ (red regions in Figure \ref{fig:cs1}(c)) and staying within state space $\X = [0,5] \times [0,5]$.

Given the PT-RAS specification, the obtained spatiotemporal tubes are shown in Figure \ref{fig:cs1}(a). Figure \ref{fig:cs1}(b) shows that the value of the time-varying CBF $b(t,x)$ is greater than $0$ for all time. 
Finally, the control effort $u^*(t,x)$ and the controlled trajectory satisfying the PT-RAS specification is shown in Figure \ref{fig:cs1}(c).

\subsection{Multi-UAV Payload System using HO-CBF}
Another significant advantage of the STT-CBF approach over the STT framework \cite{STT} is the ability to address more complicated dynamics than fully actuated affine control systems, as illustrated in the following case study. This section discusses using the STT-CBF controller for more complex systems with relative degree $\eta > 1$, such as UAVs. 

Solving PT-RAS specifications is critical for the safe and efficient operation of UAVs. For example, in drone delivery services, it is essential to ensure that the drone reaches its destination within a specified time frame while avoiding obstacles such as buildings or trees. Meeting these timing and safety constraints is vital for reliable service. 
Similarly, in a multi-UAV payload system \cite{withSTL}, where multiple drones collaborate to transport a large or delicate load, adhering to PT-RAS specifications ensures that each drone reaches specific waypoints at the correct time while avoiding collisions, and maintaining the stability of the payload. This level of coordination and precision is crucial for the success of such operations, making PT-RAS solutions indispensable in these applications. 

The latter case study of the multi-UAV payload system will be explored in detail in the following section. To establish forward invariance of the time-varying set $\Cs(t)$, for the affine system in \eqref{eq:sysdyn} with $\eta > 1$, we leverage \textit{higher-order and time-varying control barrier functions} (HO-CBF) \cite{TVCBF, HOCBF}.


\begin{definition}[High-Order Control Barrier Function]
A function $ b : \left[0, t_c\right] \times \X \rightarrow \R $ is a candidate HO-CBF of relative degree $\eta$ for system in \eqref{eq:sysdyn} if there exist differentiable class $\mathcal{K}$  functions $\alpha_i, \, i \in \{1, \ldots, m\}$ such that 
\begin{align}
    \sup_{u\in{U}} &[ \mathcal{L}^\eta_f b(t,x) +  \mathcal{L}_g \mathcal{L}_f^{\eta-1} b(t,x)u + \frac{\partial^\eta b(t,x)}{\partial t^\eta} \nonumber \\ 
    &+ \alpha_\eta(\psi_{\eta-1} (t,x)) + O(b(t,x))] \geq 0,
    \label{eq:hocbf-constraint}
\end{align}
where
\begin{align*}
    O(b(t,x)) &= \sum_{i=1}^{\eta-1} \mathcal{L}^i_f(\alpha_{\eta-i} \circ \psi_{\eta-i-1}) + \frac{\partial^i\left( \alpha_{\eta-i} \circ \psi_{\eta-i-1} \right)}{\partial t^i}, \\
    \psi_i(t,x) &= \dot{\psi}_{i-1}(t,x) + \alpha_i\left(\psi_{i-1}(t,x) \right), \forall i = \{1, \ldots, \eta\}
\end{align*}
with $\psi_0(t,x) = b(t,x)$.
\end{definition}

The feedback-linearized multi-UAV payload system has a relative degree of $\eta = 2$ \cite{withoutSTL}. Consequently, the HO-CBF $b(t,x)$, must satisfy the condition given by \eqref{eq:hocbf-constraint}, which is simplified as:
\begin{align}
    \begin{split}
    -\left(\mathcal{L}_g\mathcal{L}_fb\right)u \leq & \mathcal{L}_f^2b + \frac{\partial^2b}{\partial t^2} + 2k\mathcal{L}_fb + 2k\frac{\partial b}{\partial t} + k^2b\hspace{-0.2em}\label{eq:HOCBF_constraint_final}
    \end{split}
\end{align}
where the class $\mathcal{K}$ functions $\alpha_i(s) := ks, k \in \R^+$ for all $s\in\mathbb{R}_0^+$ and $\psi_0 = b$. The controller architecture will be similar to (\ref{eq:QP}). The control input generated from the STT-CBF controller is converted to the force input using feedback linearization relation, as mentioned in \cite{withSTL}. Once the force vectors that need to be produced by each UAV are found, the thrust vectoring control approach is used to provide those to the system \cite{withoutSTL}.

We simulate the waypoint navigation problem in a Gazebo environment \href{https://tinyurl.com/2psawnjm}{(Video Link)}, where obstacles are placed between waypoints and the prescribed time to reach each waypoint is set to 25 s. The simulation results are shown in Figure \ref{fig:demoSTTCBF_Drone}. 

\subsection{Discussion and Comparison}
To the best of the authors' knowledge, this is the first time that the construction of CBFs has been addressed for reach-avoid-stay specifications within a prescribed time frame. 

While the STT framework \cite{STT} has been used to solve PT-RAS tasks, it is limited to fully actuated control-affine systems and tends to produce excessively high control inputs, which can render it impractical for real-world applications. In contrast, the STT-CBF framework optimizes control effort, thereby achieving the PT-RAS specification with significantly less control effort, as demonstrated in our comparison of the two frameworks using a mobile robot in the first case study. The results, highlighted in Figure \ref{fig:uc}, clearly show a substantial reduction in the supremum of control effort, dropping from 13.0211 m/s to 0.7880 m/s, when using the STT-CBF approach.

\begin{figure}[t]
    \centering
    \includegraphics[width=\textwidth]{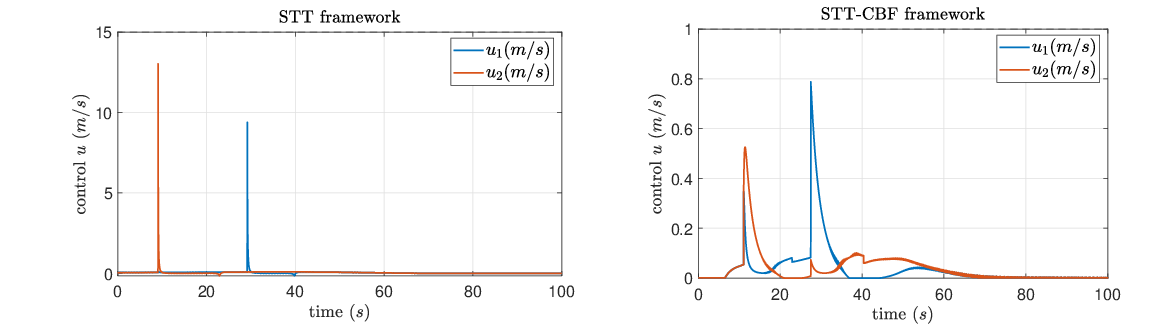}
    \caption{Comparing control effort of STT with STT-CBF.}
    \label{fig:uc}
\end{figure}

In traditional CBF-based techniques, the number of safety constraints typically scales with the number of obstacles in the environment, with each obstacle introducing an additional constraint that must be incorporated into the control law,  thus significantly increasing the computational burden and complexity of the control design \cite{learn-CBF-1}. The STT framework, on the other hand, encapsulates the entire set of obstacles within a single spatiotemporal constraint. This unified constraint reduces the computational complexity involved in computing the control input, making the STT-CBF approach advantageous in environments with a large number of obstacles.

Moreover, learning-based techniques \cite{learn-CBF-ED, learn-CBF-1, learn-CBF-2} for constructing CBFs often fall short in addressing prescribed-time performance requirements. Thus, the STT-CBF framework does not only eliminate the computational complexity associated with machine learning algorithms, providing a scalable, analytic method for deriving CBFs, but also the time-varying CBFs produced by the STT-CBF framework accurately capture the prescribed time requirements with formal correctness guarantees. Table \ref{tab:comp} accurately summarizes the advantages of the proposed STT-CBF architecture over existing algorithms in the literature. 

\begin{table}[t]
\centering
\begin{threeparttable}
\caption{Comparing STT-CBF with classical algorithms}
\begin{tabular}{lcccccccc}
\hline
\textbf{Algorithm} & \multicolumn{2}{c}{\textbf{\begin{tabular}[c]{@{}l@{}}Optimal \\ Control \end{tabular}}} & \multicolumn{2}{c}{\textbf{\begin{tabular}[c]{@{}l@{}}Formal \\ Guarantee \end{tabular}}} & \multicolumn{2}{c}{\textbf{\begin{tabular}[c]{@{}l@{}}PT-RAS \end{tabular}}} & \multicolumn{2}{c}{\textbf{\begin{tabular}[c]{@{}l@{}}Dynamics \\ with $\eta >1$ \end{tabular}}} \\ \hline
RRT*\cite{RRTs} & \multicolumn{2}{c}{-\tnote{1}} & \multicolumn{2}{c}{\xmark} & \multicolumn{2}{c}{-\tnote{2}} & \multicolumn{2}{c}{-\tnote{1}} \\
MPC\cite{MPC} & \multicolumn{2}{c}{\cmark} & \multicolumn{2}{c}{\xmark} & \multicolumn{2}{c}{\xmark} & \multicolumn{2}{c}{\xmark} \\ 
RL\cite{RLFunnel} & \multicolumn{2}{c}{\xmark} & \multicolumn{2}{c}{\xmark} & \multicolumn{2}{c}{\xmark} & \multicolumn{2}{c}{\cmark} \\
Symbolic Control\cite{tabuada2009verification} & \multicolumn{2}{c}{\xmark} & \multicolumn{2}{c}{\cmark} & \multicolumn{2}{c}{\xmark} & \multicolumn{2}{c}{\cmark} \\
STT\cite{STT} & \multicolumn{2}{c}{\xmark} & \multicolumn{2}{c}{\cmark} & \multicolumn{2}{c}{\cmark} & \multicolumn{2}{c}{\xmark} \\
CBF \cite{learn-CBF-1} & \multicolumn{2}{c}{\cmark} & \multicolumn{2}{c}{\cmark} & \multicolumn{2}{c}{\xmark} & \multicolumn{2}{c}{\cmark} \\
STT-CBF (proposed) & \multicolumn{2}{c}{\cmark} & \multicolumn{2}{c}{\cmark} & \multicolumn{2}{c}{\cmark} & \multicolumn{2}{c}{\cmark} \\
\hline
\end{tabular}

\label{tab:comp}
\begin{tablenotes}
    \item [1] Additional mechanisms like PID and MPC are required for control.
    \item [2] Additional mechanisms are required to satisfy prescribed time requirements.
\end{tablenotes}
\end{threeparttable}
\end{table}

\section{Conclusion}
In this paper, we tackle the challenge of constructing a CBF for a PT-RAS task. We utilize the STT framework, which enables the systematic design of time-varying CBFs to guarantee the satisfaction of the PT-RAS specification. Our proposed approach is validated through two case studies: a 2D omnidirectional robot path planning task, where we demonstrate the reduction in control effort using the STT-CBF framework. The second case study focuses on a temporal waypoint navigation task with obstacles for multiple UAVs, employing HO-CBFs. This study shows that the STT-CBF framework is not only applicable to more complex systems but also significantly reduces the computational complexity of CBF construction compared to data-driven approaches. These findings highlight the efficiency and scalability of the STT-CBF approach in addressing PT-RAS tasks.

\bibliographystyle{unsrt} 
\bibliography{sources} 

\begin{thebibliography}{10}

\bibitem{Meng1}
Yiming Meng, Yinan Li, and Jun Liu.
\newblock Control of nonlinear systems with reach-avoid-stay specifications: A {L}yapunov-barrier approach with an application to the {M}oore-{G}reizer model.
\newblock In {\em American Control Conference}, pages 2284--2291, 2021.

\bibitem{Kloetzer}
Marius Kloetzer and Calin Belta.
\newblock A fully automated framework for control of linear systems from temporal logic specifications.
\newblock {\em IEEE Transactions on Automatic Control}, 53(1):287--297, 2008.

\bibitem{NAHS}
Pushpak Jagtap and Dimos~V. Dimarogonas.
\newblock Controller synthesis against omega-regular specifications: {A} funnel-based control approach.
\newblock {\em International Journal of Robust and Nonlinear Control}, 34(11):7161--7173, 2023.

\bibitem{MPC}
Peng Hang, Sunan Huang, Xinbo Chen, and Kok~Kiong Tan.
\newblock Path planning of collision avoidance for unmanned ground vehicles: {A} nonlinear model predictive control approach.
\newblock {\em Proceedings of the Institution of Mechanical Engineers, Part I: Journal of Systems and Control Engineering}, 235(2):222--236, 2021.

\bibitem{APF}
Jerome Barraquand, Bruno Langlois, and J-C Latombe.
\newblock Numerical potential field techniques for robot path planning.
\newblock {\em IEEE transactions on systems, man, and cybernetics}, 22(2):224--241, 1992.

\bibitem{APF_book}
Elon Rimon.
\newblock {\em Exact robot navigation using artificial potential functions}.
\newblock Yale University, 1990.

\bibitem{tabuada2009verification}
Paulo Tabuada.
\newblock {\em Verification and control of hybrid systems: a symbolic approach}.
\newblock Springer Science \& Business Media, 2009.

\bibitem{PPC1}
Charalampos~P. Bechlioulis and George~A. Rovithakis.
\newblock Robust adaptive control of feedback linearizable {MIMO} nonlinear systems with prescribed performance.
\newblock {\em IEEE Transactions on Automatic Control}, 53(9):2090--2099, 2008.

\bibitem{Funnel_STL}
Lars Lindemann, Christos~K. Verginis, and Dimos~V. Dimarogonas.
\newblock Prescribed performance control for signal temporal logic specifications.
\newblock In {\em IEEE 56th Conference on Decision and Control}, pages 2997--3002, 2017.

\bibitem{CBF}
Aaron~D. Ames, Xiangru Xu, Jessy~W. Grizzle, and Paulo Tabuada.
\newblock Control barrier function based quadratic programs for safety critical systems.
\newblock {\em IEEE Transactions on Automatic Control}, 62(8):3861--3876, 2017.

\bibitem{jagtap2020formal}
Pushpak Jagtap, Sadegh Soudjani, and Majid Zamani.
\newblock Formal synthesis of stochastic systems via control barrier certificates.
\newblock {\em IEEE Transactions on Automatic Control}, 66(7):3097--3110, 2020.

\bibitem{CBF_TA}
Aaron~D Ames, Samuel Coogan, Magnus Egerstedt, Gennaro Notomista, Koushil Sreenath, and Paulo Tabuada.
\newblock Control barrier functions: {T}heory and applications.
\newblock In {\em 2019 18th European control conference (ECC)}, pages 3420--3431. IEEE, 2019.

\bibitem{CBF-RRT1}
Aniketh Manjunath and Quan Nguyen.
\newblock Safe and robust motion planning for dynamic robotics via control barrier functions.
\newblock In {\em 2021 60th IEEE Conference on Decision and Control (CDC)}, pages 2122--2128. IEEE, 2021.

\bibitem{CBF-RRT2}
Guang Yang, Bee Vang, Zachary Serlin, Calin Belta, and Roberto Tron.
\newblock Sampling-based motion planning via control barrier functions.
\newblock In {\em Proceedings of the 2019 3rd International Conference on Automation, Control and Robots}, pages 22--29, 2019.

\bibitem{CBF-RRT3}
Keyvan Majd, Shakiba Yaghoubi, Tomoya Yamaguchi, Bardh Hoxha, Danil Prokhorov, and Georgios Fainekos.
\newblock Safe navigation in human occupied environments using sampling and control barrier functions.
\newblock In {\em 2021 IEEE/RSJ International Conference on Intelligent Robots and Systems (IROS)}, pages 5794--5800. IEEE, 2021.

\bibitem{CBF-APF}
Andrew Singletary, Karl Klingebiel, Joseph Bourne, Andrew Browning, Phil Tokumaru, and Aaron Ames.
\newblock Comparative analysis of control barrier functions and artificial potential fields for obstacle avoidance.
\newblock In {\em 2021 IEEE/RSJ International Conference on Intelligent Robots and Systems (IROS)}, pages 8129--8136. IEEE, 2021.

\bibitem{Meng3}
Yiming Meng, Yinan Li, Maxwell Fitzsimmons, and Jun Liu.
\newblock Smooth converse {L}yapunov-barrier theorems for asymptotic stability with safety constraints and reach-avoid-stay specifications.
\newblock {\em Automatica}, 144:110478, 2022.

\bibitem{C3BF}
Bhavya~Giri Goswami, Manan Tayal, Karthik Rajgopal, Pushpak Jagtap, and Shishir Kolathaya.
\newblock Collision cone control barrier functions: {E}xperimental validation on {UGV}s for kinematic obstacle avoidance.
\newblock In {\em American Control Conference (ACC)}, 2024.

\bibitem{learn-better-CBF}
Bolun Dai, Prashanth Krishnamurthy, and Farshad Khorrami.
\newblock Learning a better control barrier function.
\newblock In {\em 2022 IEEE 61st Conference on Decision and Control (CDC)}, pages 945--950. IEEE, 2022.

\bibitem{learn-CBF-1}
Mohit Srinivasan, Amogh Dabholkar, Samuel Coogan, and Patricio~A Vela.
\newblock Synthesis of control barrier functions using a supervised machine learning approach.
\newblock In {\em 2020 IEEE/RSJ International Conference on Intelligent Robots and Systems (IROS)}, pages 7139--7145. IEEE, 2020.

\bibitem{learn-CBF-ED}
Alexander Robey, Haimin Hu, Lars Lindemann, Hanwen Zhang, Dimos~V Dimarogonas, Stephen Tu, and Nikolai Matni.
\newblock Learning control barrier functions from expert demonstrations.
\newblock In {\em 2020 59th IEEE Conference on Decision and Control (CDC)}, pages 3717--3724. IEEE, 2020.

\bibitem{TVCBF}
Xiangru Xu.
\newblock Constrained control of input–output linearizable systems using control sharing barrier functions.
\newblock {\em Automatica}, 87:195--201, 2018.

\bibitem{STT}
Ratnangshu Das and Pushpak Jagtap.
\newblock Prescribed-time reach-avoid-stay specifications for unknown systems: {A} spatiotemporal tubes approach.
\newblock {\em IEEE Control Systems Letters}, 8:946--951, 2024.

\bibitem{HOCBF}
Wei Xiao and Calin Belta.
\newblock High-order control barrier functions.
\newblock {\em IEEE Transactions on Automatic Control}, 67(7):3655--3662, July 2022.

\bibitem{CBF_STL}
Lars Lindemann and Dimos~V. Dimarogonas.
\newblock Control barrier functions for signal temporal logic tasks.
\newblock {\em IEEE Control Systems Letters}, 3(1):96--101, 2019.

\bibitem{withSTL}
Nishanth Rao, Suresh Sundaram, and Pushpak Jagtap.
\newblock Temporal waypoint navigation of multi-{UAV} payload system using barrier functions, 2022.

\bibitem{withoutSTL}
Nishanth Rao and Suresh Sundaram.
\newblock An input-output feedback linearization based exponentially stable controller for multi-{UAV} payload transport, 2022.

\bibitem{learn-CBF-2}
Lars Lindemann, Haimin Hu, Alexander Robey, Hanwen Zhang, Dimos Dimarogonas, Stephen Tu, and Nikolai Matni.
\newblock Learning hybrid control barrier functions from data.
\newblock In {\em Conference on robot learning}, pages 1351--1370. PMLR, 2021.

\bibitem{RRTs}
Sertac Karaman and Emilio Frazzoli.
\newblock Sampling-based algorithms for optimal motion planning.
\newblock {\em The International Journal of Robotics Research}, 30(7):846--894, 2011.

\bibitem{RLFunnel}
Naman Saxena, Sandeep Gorantla, and Pushpak Jagtap.
\newblock Funnel-based reward shaping for signal temporal logic tasks in reinforcement learning.
\newblock {\em IEEE Robotics and Automation Letters}, 9(2):1373--1379, 2024.

\end{thebibliography}

\end{document}